\documentclass[12pt]{article}
\usepackage{amsmath}
\usepackage{amsfonts}
\usepackage{amsthm}
\usepackage{amssymb}
\usepackage{amstext}
\usepackage[margin=1in]{geometry}
\usepackage{enumitem}
\usepackage{listings}
\usepackage{graphicx}
\usepackage{pdflscape}
\usepackage{mathtools}
\usepackage{scrextend}
\usepackage{array}
\usepackage{multicol}
\usepackage{mathrsfs}
\usepackage{color}
\usepackage{bm}
\usepackage{changebar}

\usepackage{graphicx}

\usepackage[colorlinks=true,
linkcolor=webgreen,
filecolor=webbrown,
citecolor=webgreen]{hyperref}

\definecolor{webgreen}{rgb}{0,.5,0}
\definecolor{webbrown}{rgb}{.6,0,0}

\usepackage{color}

\theoremstyle{definition}
\newtheorem{theorem}{Theorem}
\newtheorem{lemma}{Lemma}

\newcommand{\Enn}{\mathbb{N}}

\newcommand{\seqnum}[1]{\href{https://oeis.org/#1}{\rm \underline{#1}}}

\author{Robert Cummings, Jeffrey Shallit, and Paul Staadecker\\
School of Computer Science\\
University of Waterloo\\
Waterloo, ON  N2L 3G1 \\
Canada\\
\href{mailto:rcummings000@gmail.com}{\tt rcummings000@gmail.com} \\
\href{mailto:shallit@uwaterloo.ca}{\tt shallit@uwaterloo.ca}\\
\href{mailto:staadeckerpaul@gmail.com}{\tt staadeckerpaul@gmail.com}
}

\title{Mesosome Avoidance}
\begin{document}
\maketitle

\begin{abstract}
We consider avoiding mesosomes---that is, words of the form $xx'$ with
$x'$ a conjugate of $x$ that is different from $x$---over a binary alphabet.
We give a structure theorem for mesosome-avoiding words, count how many
there are, characterize all the infinite mesosome-avoiding words, and
determine the minimal forbidden words.
\end{abstract}

\section{Introduction}
\label{intro}
Since the early 20th century, with the work of Axel Thue
\cite{Thue:1906,Thue:1912}, pattern avoidance in words has been a topic of interest.   Thue proved that there exists an infinite word over a $3$-letter alphabet having no factor (i.e., a contiguous sub-block) of the form $xx$, where $x$ is a nonempty word.  We say that such a word {\it avoids squares}.  It is easy to see that every binary word of length $\geq 4$ has a square, so one cannot avoid squares over a $2$-letter alphabet.  Thue also proved that there exists an infinite word over a $2$-letter alphabet avoiding overlaps, that is, factors of the form $axaxa$, where $a$ is a single letter and $x$ is a possibly empty word. 

Since Thue's pioneering work, many similar avoidance problems have been studied.   For example, Erd\H{o}s \cite{Erdos:1961} posed the problem of avoiding abelian squares in words---i.e., words of the form $x x'$ where $x'$ is a permutation of $x$---and Ker\"anen \cite{Keranen:1992} showed that it is possible to avoid abelian squares over a $4$-letter alphabet; the alphabet size is optimal.
In \cite{Loftus&Shallit&Wang:2000}, the authors studied words avoiding $x \sigma(x)$, where $\sigma$ is a cyclic shift of the underlying alphabet.   In \cite{Ng&Ochem&Rampersad&Shallit:2019}, the authors studied words avoiding $x h(x)$, where $h$ is an arbitrary nonerasing morphism.   In \cite{Ochem&Rampersad&Shallit:2008}, the authors
studied words avoiding $x x'$ where $x'$ is ``nearly'' identical to $x$.  A classic problem, due independently to Justin \cite{Justin:1972a}, Brown and Freedman \cite[Conjecture, pp.~595--596]{Brown&Freedman:1987}, 
Pirillo and Varricchio \cite{Pirillo&Varricchio:1994},
and
Halbeisen and Hungerb\"uhler \cite{Halbeisen&Hungerbuhler:2000} and still unsolved,
asks if it is possible to avoid additive squares (words of the form $x x'$ where
$|x| = |x'|$ and $\sum x = \sum x'$) over some finite subset of 
$\Enn$, the natural numbers.

A natural pattern to consider is $x x'$, where $x'$ is a conjugate (cyclic shift) of $x$. For example, the English words {\tt enlist} and {\tt listen} are conjugates.  However, since $x$ and $x'$ are conjugate if and only if there exist words $u, v$ with $x = uv$ and $x' = vu$, we see that a word contains this kind of pattern if and only if it contains a square.   Thus, it is impossible to avoid such a pattern over a binary alphabet, although it can be avoided over a ternary alphabet.

Instead, in this note, we consider a minor variation of this pattern.   We say a finite nonempty word $w$ is a {\it mesosome\/} if it is of the form $x x'$, with $x'$ a conjugate of $x$, and $x \not= x'$.   The word {\tt mesosome} itself is a mesosome, as {\tt some} is a cyclic shift of {\tt meso}.    Every squarefree word, of course, avoids mesosomes, so mesosomes can certainly be avoided over a $3$-letter alphabet.   Our goal is to consider mesosome avoidance over a binary alphabet.

There are four principal results.  First, we characterize
all finite binary words avoiding mesosomes.   Second, we count how many such words there are of length $n$, and show this number is a cubic function of $n$. Third, we characterize all infinite binary words avoiding mesosomes.  Finally, we characterize the minimal forbidden words for mesosome-avoiding words.

\section{Characterizing the finite mesosome-avoiding words}
\label{two}

We define $\overline{0} = 1$ and $\overline{1} = 0$,
and extend this to finite and infinite words in the obvious way.   Clearly a binary word $x$ has a mesosome iff $\overline{x}$ has a mesosome.   Without loss of generality, then, in our characterization we can restrict our attention to nonempty words that begin with $0$.  We make this assumption in what follows.

The basic idea is to classify mesosome-avoiding words starting with $0$ by the number of runs they contain.  A {\it run\/} is a maximal block of contiguous identical symbols.   With the notation $x \sim x'$, we mean that $x$ and $x'$ are conjugates (possibly equal).

\bigskip

\noindent {\it One run:}   Clearly $0^i$ avoids mesosomes for all $i \geq 1$.

\bigskip

\noindent {\it Two runs:}     Clearly $0^i 1^j$ avoids mesosomes for all $i, j \geq 1$.

\bigskip

\noindent {\it Three runs:}   
\begin{lemma}
A word of the form $0^i 1^j 0^k$ 
(for $i, j, k \geq 1$) avoids mesosomes iff $j$ is odd.
\label{threeruns}
\end{lemma}

\begin{proof}
If $w = 0^i 1^j 0^k$ is mesosome-free, then $j$ must be odd. For otherwise $j = 2r $ and $w$ contains a mesosome factor of the form $01^{r}1^{r}0$.  For the converse, suppose $j$ is odd and $w = 0^i 1^j 0^k$ contains a mesosome $x x'$.   Then $x$ must start with $0$ and end with $1$ and $x'$ must start with $1$ and end with $0$.   But then $j$ would be even, a contradiction.
\end{proof}

\bigskip

\noindent{\it Four runs:}  
\begin{lemma}
The word $0^i 1^j 0^k 1^\ell$ (for $i, j, k, \ell \geq 1$) avoids mesosomes iff \begin{itemize}
    \item[(a)] $j = k = 1$ or 
    \item[(b)] $i < k$ and $j > \ell$ and $j, k$ are odd.
\end{itemize}
\label{fourruns}
\end{lemma}

\begin{proof}
Let us show that the strings in Cases (a) and (b) avoid mesosomes.

\bigskip

\noindent{\it Case (a): }  $w = 0^i 1 0 1^\ell$.   Suppose $w$ has a mesosome.
Then $w = u x x' v$ with $x \sim x'$ and $x \ne x'$.
Suppose $x$ begins inside the first group of $0$'s in $w$.   If it ends
inside the first group of $0$'s, then either $x' = x$ or $x'$ contains
a $1$, a contradiction.   So $x$ contains a $1$.    But then either
$x = 0^s 1$, forcing $x' = x$, or $x = 0^s 1 0$, forcing $x'$ to have
no $0$'s.   

So $x$ begins with the first $1$.  But then there is no possible choice
for $x'$.

So $x$ begins at a position $\geq i+2$ in $w$.   But then either $x$ starts with
$0$ and $x'$ doesn't, or $x = 1^t$, forcing $x' = x$.
In all cases we get a contradiction.

\bigskip

\noindent{\it Case (b):}  $w = 0^i 1^j 0^k 1^\ell$, with $i < k$ and $j > \ell$ and $j, k$ are odd.

Suppose $w$ has a mesosome. Then again, $w=uxx'v$ with $x \sim x'$ and $x \ne x'$. Suppose $x$ begins inside the first group of $0$'s in $w$. Then as in the previous case, $x$ cannot end in the first group of $0$'s. Suppose $x$ ends inside the first group of $1$'s in $w$. Then since $j$ is odd, $x'$ must contain some $1$'s from the second group of $1$'s. This means that $x'$ has $k$ $0$'s. Since $k > i$,  we see that $x'$ has more $0$'s than $x$ does, a contradiction. Suppose that $x$ ends after the first group of $1$'s. Then $x$ has at least $j$ $1$'s, and $x'$ has fewer than $j$ $1$'s, a contradiction.

So $x$ must begin within the first group of $1$'s or after. But then $x'$ must end in the second group of $1$'s. Since $k$ is odd and $j >\ell$, we can apply the argument above to $\overline{w^R}$. 

Now we show that these are the only mesosome-avoiding binary words
with four runs.

Suppose $w= 0^i 1^j 0^k 1^\ell$ with $i, j,k,\ell \geq 1$.
Considering the factor $0^i 1^j 0^k$ gives us that $j$ is odd, by Lemma~\ref{threeruns}, and
considering the factor $1^j 0^k 1^\ell$ gives us that $k$ is odd, also by
Lemma~\ref{threeruns}.  

There are four cases to consider:
\bigskip

\noindent{\it Case 1:}  $i \geq k$, $j > 1$:
Write $w = 0^{i-k} x x' 1^{\ell - 1}$, with
$x = 0^k 1^{(j+1)/2}$ and $x' = 1^{(j-1)/2} 0^k 1$.
Note that $x \sim x'$ and $x \not= x'$, a contradiction.

\bigskip

\noindent{\it Case 2:}  $i \geq k$, $j = 1$, $k > 1$:
Write $w = 0^{i-1} x x' 1^{\ell - 1}$, with
$x = 0 1 0^{(k-1)/2}$ and $x' = 0^{(k+1)/2} 1$.
Note that $x \sim x'$ and $x \not= x'$, a contradiction.

\bigskip

\noindent{\it Case 3:}   $j \leq \ell$, $k > 1$. \\
\noindent{\it Case 4:}  $j \leq \ell$, $k = 1$, $j > 1$.
Cases 3 and 4 are completely parallel to Cases 1 and 2, by considering $\overline{w^R}$
instead of $w$.
\end{proof}

\bigskip

\noindent{\it Five runs:}   
\begin{lemma}
If $x$ is mesosome-free and has 5 runs, then $x = 01010$.
\label{fiveruns}
\end{lemma}

\begin{proof}
Let $x = 0^i 1^j 0^k 1^\ell 0^m$.   
We can now apply Lemma~\ref{fourruns} to the first four and last four runs.
This gives four possibilities:

\bigskip

\noindent{\it Case 1:}
$j=k=\ell = 1$.   If $i \geq 2$, then $w = 0^{i-2} x x' 0^{m-1}$,
with $x = 001$ and $x' =01 0$.   So $x \sim x'$ and $x \not= x'$.
The case $m \geq 2$ follows analogously by considering $\overline{w^R}$.
So $i = m = 1$.  In this case $w = 01010$, which has no mesosome.

\bigskip

\noindent{\it Case 2:}  $j = k = 1$ and $j < \ell$ and $k > m$.   But then
$k = 1 > m \geq 1$, a contradiction.

\bigskip

\noindent{\it Case 3:}  $i< k$ and $j > \ell$ and $k = \ell = 1$.   But then
$i < k = 1$, a contradiction.

\bigskip

\noindent{\it Case 4:}  $i < k$ and $j > \ell$ and $j < \ell$ and $k > m$.
But then $j > \ell$ and $j < \ell$, a contradiction.
\end{proof}

\bigskip
\noindent{\it More than $5$ runs:}
By applying Lemma~\ref{fiveruns} to $w$, we see that the first five runs
must be $01010$.   We can then apply the same lemma to each consecutive
group of $5$ runs and conclude that $w = (01)^{n/2}$ or $w = (01)^{(n-1)/2} 0$.

For the converse, we need to see that these words are mesosome-free.
But this is clear, since any factor of the form $y y'$ with
$|y|=|y'|$ either satisfies $y = y'$ if $|y|$ is even, or
$y = \overline{y'}$ if $|y|$ is odd.  In this latter case, $y'$ contains
one more copy of one letter than $y$ does.

Thus we have proved the following result:
\begin{theorem}
  The finite binary word $x$ is mesosome-free iff $x$ or $\overline{x}$ has one of the following forms:
  \begin{itemize}
    \item[(a)] $0^{i}$, $i \geq 1$;
    \item[(b)] $0^{i}1^{j}$, $i, j \geq 1$;
    \item[(c)] $0^{i}1^{j}0^{k}$ for $i, j, k \geq 1$ and $j$ odd;
    \item[(d)] $0^{i}101^{j}$ for $i, j \geq 1$;
    \item[(e)] $0^{i}1^{j}0^{k}1^{\ell}$ for $i, j, k,\ell \geq 1$
    and $j,k$ odd, $i<k$ and $j>l$;
    \item[(f)] $(01)^{i}$ or  $(01)^{i}0$ for $i \geq 1$.
  \end{itemize}
  \label{Theorem 1}
\end{theorem}

\section{Counting the mesosome-avoiding words}

Based on the characterization of Section~\ref{two}, we can now count the mesosome-avoiding binary words of length $n$.   
The following table gives the number $m_n$ of such words of length $n$ for the first
few values of $n$:
\begin{center}
\begin{tabular}{c|ccccccccccccccccccccc}
$n$ & 0& 1& 2& 3& 4& 5& 6& 7& 8& 9&10&11&12&13&14&15&16&17\\
\hline
$m_n$ & 1&  2&  4&  8& 14& 24& 32& 42& 54& 68& 82& 98&118&140&162&186&216&248
\end{tabular}
\end{center}
It is sequence \seqnum{A341277} in the {\it On-Line Encyclopedia of Integer Sequences} \cite{Sloane}.

\begin{theorem}
Let $m(n)$ denote the number of mesosome-avoiding binary words of length $n$.
Write $n = 4k+i$ for $i = 0, 1, 2, 3$ and $n \geq 5$.
Then
\begin{align}
m(4k) &= (4k^3 + 15k^2 +41k -12)/3; \nonumber \\
m(4k+1) &= (4k^3 + 18k^2 + 50k)/3; \label{mesoformula}\\
m(4k+2) &= (4k^3 + 21k^2 + 59k + 12)/3; \nonumber\\
m(4k+3) &= (4k^3 + 24k^2 + 68k + 30)/3. \nonumber
\end{align}
\end{theorem}

\begin{proof}
In what follows we only consider
the words that begin with $0$; to get the total of all such words, as provided in the statement of the the theorem, it is necessary to multiply the intermediate results below by $2$.

\bigskip

\noindent{\it Words with one run:}   There is only
one, namely $0^n$.

\bigskip

\noindent{\it Words with two runs:}  We must count
the words of the form $0^i 1^j$ for $i, j \geq 1$
and $i+j = n$.   Clearly there are $n-1$ such words.

\bigskip

\noindent{\it Words with three runs:}  We must
count words of the form $0^i 1^j 0^k$ for $i,j,k \geq 1$
and $j$ odd.  Once $j$ is fixed, there are clearly
$n-j-1$ such words, so letting $j = 2j'-1$, we see that the total is
$$\sum_{1 \leq j' \leq (n-1)/2} n-2j' = 
\begin{cases} n(n-2)/4, & \text{if $n$ is even}; \\
(n-1)^2/4, & \text{if $n$ is odd}.
\end{cases}$$

\noindent{\it Words with four runs:}   Here there
are two cases.   Either $w = 0^i 1 0 1^j$ for $i, j \geq 1$, or $w = 0^i 1^j 0^k 1^\ell$ with $i,j,k,\ell \geq 1$,
$i < k$, $j > \ell$, and $j, k$ odd.

In the former case there are clearly $n-3$ such words.

To count the second case, let $f(r)$ denote the
number of ways to write $r$ as the sum $r = i + k$
with $i, k \geq 1$ and $i < k$ and $k$ odd.   Then it is easy to see
that $f(r) = \lfloor r/4 \rfloor$.   Then the second
case is clearly $g(n) = \sum_{0 \leq r \leq n} f(r)f(n-r)$.   To compute $g(n)$, rewrite
$f(r) = \lfloor r/4 \rfloor$ as 
$$r/4  - 3/8 + (-1)^r/8 + (1-i)i^r/8 + (1+i)(-i)^r/8,$$ 
and then sum $g$ in the standard way (or use Maple).
We get
$$g(n) = {{n-3} \over 96} (n^2 - 6n + 3(-1)^n/2 + 1/2) +
{1 \over {192}} ((6in -18i + 6(-1)^n)(-i)^n +
(-6in + 18i + 6(-1)^n)i^n) .$$
Considering the possibilities mod 4, this gives
\begin{align*}
g(4n) &= 2n^3/3 - (3/2)n^2 + 5n/6 \\
g(4n+1) &= 2n^3/3 - n^2 + n/3 \\
g(4n+2) &= 2n^3/3 - n^2/2 - n/6 \\
g(4n+3) &= 2n^3/3 -2n/3.
\end{align*}

\noindent{\it Words with five or more runs:}  From Lemma~\ref{fiveruns} we see that for $n \geq 5$,
there is exactly one such word of length $n$.

\bigskip

Now, adding up all five cases, and multiplying by $2$, we
get that the total number $m(n)$ of mesosome-avoiding words
of length $n \geq 5$ is given by 
Eqn.~\eqref{mesoformula}.

\end{proof}

\section{Infinite binary words avoiding mesosomes}

Recall that by $x^\omega$, for a finite
word $x$, we mean the infinite word
$xxx\cdots$.

\begin{theorem}
    Let $i, j \geq 1$.
  If $\bf x$ is an infinite binary mesosome-free word, then $\bf x$
  has one of the following forms (or their binary complement):
  \begin{itemize}
    \item $0^{\omega}$
    \item $0^{i} \, 1^{\omega}$
    \item $0^{i} \, 1^{2j-1} \, 0^{\omega}$
    \item $0^{i}10 \, 1^{\omega}$
    \item $(01)^{\omega}$
  \end{itemize}
\end{theorem}
\begin{proof}
First suppose that $\bf x$ is an infinite binary mesosome-free word. Let $x_n$ denote the first $n$ letters of $\bf x$, for all $n\ge 1$. Note that $x_n$ must be mesosome-free for all such $n$. This means that for all such $n$, $x_n$ must be of the form described in Theorem~\ref{Theorem 1}. 

Now suppose to the contrary that for some $n$, $x_n$ has the form $0^i1^j0^k1^\ell$ for $j,k$ odd, $i<k$ and $j>\ell$. Then if we add zeros at the end of $x_n$, we will have a word with five runs not of the form $(01)^i$ or $(01)^i0$, and so we will have a mesosome. If we add ones at the end of $x_n$, then eventually $\ell > j$ and our new word will have a mesosome. Therefore, there is no way to add letters at the end of $x_n$ while ensuring that the new word is mesosome-free. Thus, $x$ must contain a mesosome, a contradiction.

It follows that $x_n$ is of the form $0^i$, $0^i1^j$, $0^i1^{2j-1}0^\ell$, $0^i101^j$, $(01)^i$, or $(01)^i0$ for all $n$. It is easy to see by letting $n$ approach infinity that $x$ must be of the required form.

Now suppose that $x$ is of the required form. Then clearly for all natural $n$, $x_n$ is mesosome-free by Theorem~\ref{Theorem 1}, so $x$ is mesosome-free.
\end{proof}

\section{Minimal forbidden words}

A word is {\it minimal forbidden} if it is a mesosome,
but no proper subword is a mesosome.   For example, the smallest minimal forbidden words are $0110$ and $1001$. There are ten minimal forbidden words
of length $6$, which are $001010$,
 $010001$,
 $010100$,
 $011101$, 
 $011110$, and their binary complements.   These are the only minimal
 forbidden words of length $\leq 7$.
Note that a minimal forbidden word is always of the form $xx'$ for $x\sim x'$, so a minimal forbidden word always has even length.

Our first goal is to characterize the minimal forbidden words.
We do this by considering  cases depending on the number of runs that the words contain.
\bigskip

\noindent{\it One or two runs:}
All words with one or two runs are mesosome-free, so none are minimal-forbidden.
\bigskip

\noindent{\it Three runs:}
Let  $w=0^i1^j0^k$ be minimal forbidden. Then $w$ is a mesosome, so $j$ must be even by Lemma~\ref{threeruns}. Now, if $i>1$, we can write $w=0x$, where $x=0^{i-1}1^j0^k$ is a mesosome, implying that $w$ is not minimal forbidden. Similarly, if $k>1$, then $w$ is not minimal forbidden. Thus, we must have $i=1=k$, and so $w=01^{n-2}0$ for even $n$.

\bigskip

\noindent{\it Four runs:}
\begin{lemma}
A word of the form $0^i1^j0^k0^\ell$ is minimal forbidden iff $j$ and $k$ are odd and either \begin{itemize}
\item[(a)] $j>1$, $\ell=1$ and $i=k$ or 
\item[(b)] $k>1$, $i=1$ and $j=\ell$.
\end{itemize}
\label{min forbidden fourruns}
\end{lemma}

\begin{proof}
Let $w=0^i1^j0^k1^\ell$ be minimal forbidden. Then $j$ and $k$ must be odd, because otherwise either $01^j0$ or $10^k1$ would be a proper subword and a mesosome. Having established that $j$ and $k$ are odd, we split into four cases as follows:
\bigskip

\noindent{\it Case 1:} $j\le \ell$ and $k \le i$. If $j,k = 1$ then $w=0^i101^\ell$ is mesosome-free by Lemma~\ref{fourruns}, so is not minimal forbidden. 
Otherwise, suppose that $j>1$. Then $\ell>1$, so we can write $w=x1^{\ell-1}$ where $x=0^i1^j0^k1$ is a proper subword and (since $k\le i$ and $j\ne 1$) a mesosome by Lemma~\ref{fourruns}. A similar argument applied to $\overline{w^R}$ will show that if $k>1$, then $w$ is still not minimal forbidden, a contradiction.
\bigskip

\noindent{\it Case 2:} $j>\ell$ and $k \le i$. Note that if $x = 0^k1^{(j+1)/2}$ and $x' = 1^{(j-1)/2}0^k1$, then $xx'$ is minimal forbidden and $w=0^{i-k}xx'1^{\ell-1}$. It follows that in this case, $w$ is minimal forbidden iff $0^{i-k}$ and $1^{\ell-1}$ are empty; that is, if and only if $i=k$ and $\ell=1$.
\bigskip

\noindent{\it Case 3:} $j \le \ell$ and $k > i$. By considering $\overline{w^R}$, we see that this case is symmetrical to case 2. It follows that in this case, $w$ is minimal forbidden iff $\ell= j$ and $i=1$.
\bigskip

\noindent{\it Case 4:} $j > \ell$ and $k > i$. By Lemma~\ref{fourruns}, $w$ is mesosome-free, so cannot be minimal forbidden.
\end{proof}
\bigskip

\noindent{\it Five runs:}
\begin{lemma}
No word of length $n\ge 8$ of the form $0^i1^j0^k1^\ell0^m$ is minimal forbidden.
\label{min forbidden fiveruns}
\end{lemma}

\begin{proof}
Suppose to the contrary that $w=0^i1^j0^k1^\ell0^m$ has length $\ge 8$ and is 
minimal forbidden. Then clearly at least one of $i$, $k$, $j$, $\ell$, and $m$ must be greater than $1$. We split into two cases based on this.
\bigskip

\noindent{\it Case 1:} $j,k,l=1$. Then $i>1$ or $m>1$. Suppose that $i>1$. Then we can write $w=0^{i-2}xx'0^{m-1}$, where $x=001$ and $x'=010$. Note that $xx'$ is a mesosome of length $6$, and since $w$ has length $\ge 8$, it must be that $xx'$ is a proper subword of $w$, a contradiction. The case where $m>1$ is symmetrical.
\bigskip

\noindent{\it Case 2:} at least one of $j$, $k$, and $l$ is different than one. Then as in the proof of Lemma~\ref{fiveruns}, it can be shown that either $0^i1^j0^k1^\ell$ or $1^j0^k1^\ell0^m$ is a mesosome, implying that $w$ is not minimal forbidden.
\end{proof}
\bigskip

\noindent{\it Six runs or more:} Suppose to the contrary that $w$ is minimal forbidden with six runs or more. Then $w$ must have at least one run of length $>1$, because otherwise $w$ would be mesosome-free by Theorem~\ref{Theorem 1}. Therefore, we can write $w=uxv$, where $x$ has five runs, with one run of length $>1$. Note that $u$ and $v$ cannot both be empty, because $w$ needs to have at least six runs. Thus, $x$ is a proper subword of $w$. Since $x$ has at least one run of length $>1$, $x$ must also be a mesosome by Theorem~\ref{Theorem 1}, a contradiction. By contradiction, we see that there are no minimal forbidden words with six runs or more.
\bigskip

We have now proved the following result:
\begin{theorem}
For $n \equiv 0$ (mod 4) and $n \geq 8$ there are $n-2$ minimal forbidden
words and they are as follows:
\begin{itemize}
\item $0^{2i-1}1^{n-4i+1}0^{2i-1}1$
for $1 \leq i \leq n/4 -1$; 
\item $0 \, 1^{2i-1}\, 0^{n-4i+1}\, 1^{2i-1}$ for $1 \leq i \leq n/4-1$;
\item $0\, 1^{n-2} \, 0$,
\end{itemize}
and their binary complements.

For
$n \equiv 2$ (mod 4) and $n \geq 10$ there are $n$ minimal forbidden
words and they are as follows:
\begin{itemize}
\item $0^{2i-1}1^{n-4i+1}0^{2i-1}1$ for
$1 \leq i \leq (n-2)/4$; 
\item $0\, 1^{2i-1} \, 0^{n-4i+1} \, 1^{2i-1}$ for $1 \leq i \leq (n-2)/4$;
\item $0 \, 1^{n-2} \, 0$,
\end{itemize} and their binary complements.
\end{theorem}

\end{document}